\documentclass[a4paper]{article}

\usepackage[dvipdfmx]{hyperref}
\usepackage[dvipdfmx]{graphicx,xcolor}
\usepackage{amsmath,amsfonts,amssymb,amsthm}

\title{Linear Complexity of Geometric Sequences Defined by Cyclotomic Classes and
Balanced Binary Sequences Constructed by the Geometric Sequences
\footnote{A preliminary version \cite{Tsuchiya-Ogawa-Nogami-Uehara} of this paper was presented at
the Eighth International Workshop on Signal Design and its Applications in Communications (IWSDA 2017).}
}
\author{Kazuyoshi Tsuchiya \footnote{The author is with the Koden Electronics Co., Ltd., Ota-ku, 146-0095 Japan} \and
Chiaki Ogawa \footnote{The authors are with the Graduate School of Natural Science and Technology, Okayama University,
Okayama-shi, 700-8530 Japan}  \and
Yasuyuki Nogami \footnotemark[3]  \and
Satoshi Uehara \footnote{The author is with the Faculty of Environmental Engineering, the University of Kitakyushu,
Kitakyushu-shi, 808-0135 Japan}
\date{}}

\newtheorem{theorem}{Theorem}
\newtheorem{proposition}[theorem]{Proposition}
\newtheorem{lemma}[theorem]{Lemma}
\newtheorem{corollary}[theorem]{Corollary}

\newtheorem{example}[theorem]{Example}
\newtheorem{remark}[theorem]{Remark}

\begin{document}
\maketitle
\begin{abstract}
Pseudorandom number generators are required to generate pseudorandom numbers
which have good statistical properties as well as unpredictability in cryptography.
An m-sequence is a linear feedback shift register sequence with maximal period over a finite field.
M-sequences have good statistical properties, however we must nonlinearize m-sequences for cryptographic purposes.
A geometric sequence is a sequence given by applying a nonlinear feedforward function to an m-sequence.
Nogami, Tada and Uehara proposed a geometric sequence
whose nonlinear feedforward function is given by the Legendre symbol,
and showed the period, periodic autocorrelation and linear complexity of the sequence.
Furthermore, Nogami et al. proposed a generalization of the sequence, and showed the period and periodic autocorrelation.
In this paper, we first investigate linear complexity of the geometric sequences.
In the case that the Chan--Games formula which describes linear complexity of geometric sequences does not hold,
we show the new formula by considering the sequence of complement numbers, Hasse derivative and cyclotomic classes.
Under some conditions, we can ensure that the geometric sequences have a large linear complexity
from the results on linear complexity of Sidel'nikov sequences.
The geometric sequences have a long period and large linear complexity under some conditions,
however they do not have the balance property.
In order to construct sequences that have the balance property,
we propose interleaved sequences of the geometric sequence and its complement.
Furthermore, we show the periodic autocorrelation and linear complexity of the proposed sequences.
The proposed sequences have the balance property,
and have a large linear complexity if the geometric sequences have a large one.
\end{abstract}
%\begin{keywords}
%pseudorandom number generator, geometric sequence, balance property, interleaved sequence.
%\end{keywords}

\section{Introduction}

Pseudorandom number generators play an important role to ensure that cryptography is safely used.
Therefore, pseudorandom number generators are required to generate pseudorandom numbers
which have good statistical properties as well as unpredictability in cryptography.

An m-sequence \cite{Golomb,Goresky-Klapper} is well-distributed sequence with long period.
However, we must nonlinearize this sequence for cryptographic purposes.
A geometric sequence \cite{Chan-Games} is a sequence given by applying a nonlinear feedforward function to an m-sequence,
such as GMW sequences \cite{Gordon-Mills-Welch,Scholtz-Welch} and cascaded GMW sequences \cite{Klapper-Chan-Goresky,Gong1996}.

Nogami, Tada and Uehara \cite{Nogami-Tada-Uehara} proposed a geometric sequence
whose nonlinear feedforward function is given by the Legendre symbol
(this sequence is referred to as the NTU sequence).
They showed the period, periodic autocorrelation and linear complexity of the NTU sequence.
The studies on generalization of the sequences have made progress in recent years
\cite{Arshad-Nogami-Ogawa-Ino-Uehara-Morelos-Zaragoza-Tsuchiya,Ogawa-Arshad-Nogami-Uehara-Tsuchiya-Morelos-Zaragoza,Nogami-Uehara-Tsuchiya-Begum-Ino-Morelos-Zaragoza,Nasima-Nogami-Uehara-Moleros-Zaragoza,Kodera-Miyazaki-Khandaker-Arshad-Nogami-Uehara,Arshad-Miyazaki-Nogami-Uehara-Morelos-Zaragoza,Heguri-Nogami-Uehara-Tsuchiya,Arshad-Miyazaki-Heguri-Nogami-Uehara-Morelos-Zaragoza,Tsuchiya-Nogami-Uehara}.
In particular, Nogami et al. \cite{Nogami-Uehara-Tsuchiya-Begum-Ino-Morelos-Zaragoza} proposed a generalization of
the NTU sequence (this sequence is referred to as the generalized NTU sequence),
and showed the period and periodic autocorrelation.

In this paper, we first investigate the linear complexity of the generalized NTU sequences.
Chan and Games \cite{Chan-Games} obtained a formula which describes the linear complexity of geometric sequences
by that of sequences with a shorter period.
However, under certain conditions, we can not apply the formula to the generalized NTU sequences.
In these cases, we show a similar formula.
This formula induces some conditions to have a large linear complexity
from the results on linear complexity of Sidel'nikov sequences.

Generalized NTU sequences have a long period and large linear complexity under some conditions,
however they do not have the balance property.
In order to construct sequences that have the balance property,
we propose interleaved sequences of the generalized NTU sequence and its complement
as in the case of NTU sequences \cite{Tsuchiya-Nogami-Uehara}.
Furthermore, we show the periodic autocorrelation and linear complexity of the interleaved sequences.
The interleaved sequences have the balance property by the interleaved structure.

This paper is organized as follows:
Sect.\,\ref{section:Preliminaries} describes some preliminaries.
Sect.\,\ref{section:Generalized NTU Sequences} introduces the generalized NTU sequences.
Sect.\,\ref{section:Linear Complexity of Generalized NTU Sequences} considers
the linear complexity of the generalized NTU sequences.
Sect.\,\ref{section:Interleaved Sequences of Generalized NTU Sequences} proposes the interleaved sequences of
the generalized NTU sequence and its complement and
show the period, periodic autocorrelation and linear complexity of the proposed sequences.
Finally, Sect.\,\ref{section:Conclusions} concludes the paper.

We give some notations.
For a prime number $p$ and an integer $a$, $(a / p)$ denotes the Legendre symbol.
For a prime power $q$, $\mathbb{F}_q$ denotes the finite field with $q$ elements.
For the extension field $\mathbb{F}_{q^m}$ of degree $m$ of $\mathbb{F}_q$,
$\mathrm{Tr}_{\mathbb{F}_{q^m} / \mathbb{F}_q} : \mathbb{F}_{q^m} \rightarrow \mathbb{F}_q$ denotes the trace map,
namely, $\mathrm{Tr}_{\mathbb{F}_{q^m} / \mathbb{F}_q} ( \alpha ) = \alpha + \alpha^q + \cdots + \alpha^{q^{m - 1}}
\in \mathbb{F}_q$ for $\alpha \in \mathbb{F}_{q^m}$.
Let $D$ be a subset of $\mathbb{F}_q$.
For $c \in \mathbb{F}_q$, we define $D + c = \{ d + c \mid d \in D \}$ and $c D = \{ c d \mid d \in D \}$.
Let $f ( x ) \in \mathbb{F}_q [ x ]$ be a polynomial over $\mathbb{F}_q$.
$\mathrm{deg} \, f$ denotes the degree of $f ( x )$.
For $\xi \in \overline{\mathbb{F}_q}$,
$\mathrm{ml}_{\xi} \, ( f ( x ) )$ denotes the multiplicity of $\xi$ as zero of $f ( x )$,
where $\overline{\mathbb{F}_q}$ is the algebraic closure of $\mathbb{F}_q$.
For an integer $n > 1$ and an integer $a$, $a \mathrm{~mod~} n$ (resp. $a ~\underline{\mathrm{mod}}~ n$)
denotes the integer $u \in \{ 0, 1, \dots ,n - 1 \}$ (resp. $u \in \{ 1, 2, \dots ,n \}$)
such that $u \equiv a \mathrm{~mod~} n$.

\section{Preliminaries}
\label{section:Preliminaries}

In this section, we introduce some definitions and known results used throughout the paper.

\subsection{Linear Complexity}

Let $\ell$ be a prime number. Let $S = ( S_n )_{n \geq 0}$ be a sequence over $\mathbb{F}_{\ell}$.
The linear complexity $L ( S )$ of $S$ is the length $L$ of the shortest linear recurrence relation
\begin{equation*}
S_{n + L} = a_{L - 1} S_{n + L - 1} + \cdots + a_0 S_n, \quad n \geq 0
\end{equation*}
for some $a_0, \dots ,a_{L - 1} \in \mathbb{F}_{\ell}$,
and the minimal polynomial $m_S ( x )$ of $S$ is the polynomial
$x^L - a_{L - 1} x^{L - 1} - \cdots - a_0  \in \mathbb{F}_{\ell} [ x ]$. 
Assume that $S$ is a periodic sequence of period $N$.
$S ( x )$ denotes the polynomial $S_0 + S_1 x + \cdots + S_{N - 1} x^{N - 1} \in \mathbb{F}_{\ell} [ x ]$.
Then $m_S ( x ) = ( x^N - 1 ) / \mathrm{gcd} \, ( x^N - 1, S ( x ) )$ and
$L ( S ) = N - \deg \mathrm{gcd} \, ( x^N - 1, S ( x ) )$ (cf. \cite[Lemma 8.2.1]{Cusick-Ding-Renvall}).

\subsection{Periodic Cross-Correlation and Periodic Autocorrelation}

Let $S^{( 1 )} = ( S^{( 1 )}_n )_{n \geq 0}$ and $S^{( 2 )} = ( S^{( 2 )}_n )_{n \geq 0}$
be periodic sequences of period $N$ over $\mathbb{F}_2$.
For $\tau \in \{ 0, 1, \dots , N - 1 \}$, $R_{S^{( 1 )}, S^{( 2 )}} ( \tau )$ is defined as
\begin{equation*}
R_{S^{( 1 )}, S^{( 2 )}} ( \tau ) = \sum_{i = 0}^{N - 1} ( - 1 )^{S^{( 1 )}_i + S^{( 2 )}_{i + \tau}}.
\end{equation*}
Note that $S^{( 1 )}_i + S^{( 2 )}_{i + \tau}$ means an addition in $\mathbb{F}_2$
and $i + \tau$ is performed modulo $N$.
For an integer $\tau \not\in \{ 0, 1, \dots , N - 1 \}$,
$R_{S^{( 1 )}, S^{( 2 )}} ( \tau )$ means $R_{S^{( 1 )}, S^{( 2 )}} ( \tau \mathrm{~mod~} N )$ in this paper.
If $S^{( 1 )} \neq S^{( 2 )}$,
then $R_{S^{( 1 )}, S^{( 2 )}} ( \tau )$ is called a periodic cross-correlation of $S^{( 1 )}$ and $S^{( 2 )}$.
If $S := S^{( 1 )} = S^{( 2 )}$, then $R_S ( \tau ) := R_{S, S} ( \tau )$ is called a periodic autocorrelation of $S$.
By the definition,  $R_S ( 0 ) = N$.

\subsection{Geometric Sequences}

Let $q$ be an odd prime power, $m > 1$ an integer, $\omega$ a primitive element in $\mathbb{F}_{q^m}$ and
$R = ( R_n )_{n \geq 0}$ an m-sequence defined as
$R_n = \mathrm{Tr}_{\mathbb{F}_{q^m} / \mathbb{F}_q} ( \omega^n ), \, n \geq 0$.
Put $\nu = ( q^m - 1 ) / ( q - 1 )$ and $g = \omega^{\nu} \in \mathbb{F}_q$.
Let $\ell$ be a prime divisor of $q - 1$ and $\rho : \mathbb{F}_q \rightarrow \mathbb{F}_{\ell}$ a (nonlinear) map.
Let $S = ( S_n )_{n \geq 0}$ be a sequence defined as $S_n = \rho ( R_n ), \, n \geq 0$,
and $s = ( s_n )_{n \geq 0}$ a sequence defined as $s_n = \rho ( g^n ), \, n \geq 0$.
$S$ is called a geometric sequence.
Chan and Games \cite{Chan-Games} showed that the linear complexity of $S$ is described by that of $s$
under certain condition as follows:

\begin{lemma}[\cite{Chan-Games} Theorem]
Assume that $\rho ( 0 ) = 0$. Then
\begin{equation}
L ( S ) = \nu L ( s ).
\label{equation:Chan-Games}
\end{equation}
\label{lemma:Chan-Games}
\end{lemma}

Chan and Games showed (\ref{equation:Chan-Games}) in the case of $\ell = 2$.
In fact, (\ref{equation:Chan-Games}) holds for any divisor $\ell$ of $q - 1$.

\subsection{Cyclotomic Classes}

We recall the definition of cyclotomic classes and difference parameters.
For details, see \cite[13.3]{Cusick-Ding-Renvall}.
Let $p$ be a prime number and $d$ a divisor of $p - 1$. Let $g$ be a primitive element in $\mathbb{F}_p$. The sets
\begin{equation*}
D_0 = \left\{ g^{i d} \, \middle| \, 0 \leq i < \frac{p - 1}{d} \right\}, \,
D_j = g^j D_0, \, 1 \leq j \leq d - 1
\end{equation*}
are called cyclotomic classes of order $d$.
For $a \in \mathbb{F}_p$, the difference parameter is defined as
$d ( i, j ; a ) = \# ( D_i \cap ( D_j - a ) ), \, i, j \in \{ 0, 1, \dots , d - 1 \}$.
If $a = 1$, then $d ( i, j ; 1 )$ is the cyclotomic number $( i, j )_d$ of order $d$.

Assume that $d = 2$.
Then $D_0$ and $D_1$ are the sets of non-zero quadratic residues modulo $p$
and of quadratic non-residues modulo $p$, respectively.
For $a \in \mathbb{F}_p$, $d ( i, j ; a )$ is given as in Table \ref{table:diff param}.

\begin{table}[h]
\begin{center}
\caption{The difference parameter $d ( i, j ; a )$ of order $2$}   \label{table:diff param}
{\renewcommand\arraystretch{1.5}
\begin{tabular}{ccccc}
\hline\noalign{\smallskip}
& \multicolumn{2}{c}{$p \equiv 1 \mathrm{~mod~} 4$} & \multicolumn{2}{c}{$p \equiv 3 \mathrm{~mod~} 4$} \\
$( i, j )$ & $( a / p ) = 1$ & $( a / p ) = - 1$ & $( a / p ) = 1$ & $( a / p ) = - 1$ \\
\noalign{\smallskip}\hline\noalign{\smallskip}
$( 0, 0 )$ & $( p - 5 ) / 4$ & $( p - 1 ) / 4$ & $( p - 3 ) / 4$ & $( p - 3 ) / 4$ \\
$( 0, 1 )$ & $( p - 1 ) / 4$ & $( p - 1 ) / 4$ & $( p + 1 ) / 4$ & $( p - 3 ) / 4$ \\
$( 1, 0 )$ & $( p - 1 ) / 4$ & $( p - 1 ) / 4$ & $( p - 3 ) / 4$ & $( p + 1 ) / 4$ \\
$( 1, 1 )$ & $( p - 1 ) / 4$ & $( p - 5 ) / 4$ & $( p - 3 ) / 4$ & $( p - 3 ) / 4$ \\
\noalign{\smallskip}\hline
\end{tabular}}
\end{center}
\end{table}

\subsection{Hasse Derivative}

Let $\ell$ be a prime number.
For $S ( x ) = \sum_{i = 0}^{N - 1} a_i x^i \in \mathbb{F}_{\ell} [ x ]$ and $k \geq 1 \in \mathbb{Z}$,
\begin{equation*}
S^{( k )} ( x ) = \sum_{i = k}^{N - 1} \binom{i}{k} a_i x^{i - k}
\end{equation*}
is called the $k$th Hasse derivative of $S ( x )$.
Let $\xi \in \overline{\mathbb{F}_{\ell}}$ be a root of $S ( x )$,
where $\overline{\mathbb{F}_{\ell}}$ is the algebraic closure of $\mathbb{F}_{\ell}$.
Then $\mathrm{ml}_{\xi} \, S ( x ) = v$ if and only if
$S ( \xi ) = S^{( 1 )} ( \xi ) = \cdots = S^{( v - 1 )} ( \xi ) = 0$ and $S^{( v )} ( \xi ) \neq 0$
(cf. \cite[Lemma 6.51]{Lidl-Niederreiter}).

\subsection{Interleaved Sequences and Left Cyclic Shift Sequences}

For a family $\mathbb{S} = \{ S^{( i )} = ( S^{( i )}_n )_{n \geq 0} \mid 0 \leq i \leq T - 1 \}$
of periodic sequences of period $N$, the sequence $U = ( U_n )_{n \geq 0}$ defined as
\begin{equation*}
U_n = S^{( i )}_j \mbox{~if~} n = j \times T + i, \, 0 \leq i \leq T - 1, \, 0 \leq j
\end{equation*}
is called an interleaved sequence of $\mathbb{S}$.

For an periodic sequence $S = ( S_n )_{n \geq 0}$ of period $N$ and $e \in \{ 0, 1, \dots , N - 1 \}$,
the sequence $L^e ( S ) = ( L^e ( S )_n )_{n \geq 0}$ defined as
\begin{equation*}
L^e ( S )_n = S_{n + e}, \quad n \geq 0
\end{equation*}
is called a left cyclic shift sequence.

\section{Generalized NTU Sequences}
\label{section:Generalized NTU Sequences}

In this section, we survey properties of the generalized NTU sequences. For details,
see \cite{Nogami-Tada-Uehara,Nasima-Nogami-Uehara-Moleros-Zaragoza,Nogami-Uehara-Tsuchiya-Begum-Ino-Morelos-Zaragoza}.

Let $p$ be an odd prime number, $m > 1$ an integer, $\omega$ a primitive element in $\mathbb{F}_{p^m}$ and
$R = ( R_n )_{n \geq 0}$ an m-sequence defined as
$R_n = \mathrm{Tr}_{\mathbb{F}_{p^m} / \mathbb{F}_p} ( \omega^n ), \, n \geq 0$.
Put $\nu = ( p^m - 1 ) / ( p - 1 )$ and $g = \omega^{\nu} \in \mathbb{F}_p$.
Let $\ell$ be a prime divisor of $p - 1$ and $D_j, \, 0 \leq j \leq \ell - 1$ cyclotomic classes of order $\ell$.
Let $A \in \mathbb{F}_p$. We define a map $\rho_{A}^{\mathrm{NTU}} : \mathbb{F}_p \rightarrow \mathbb{F}_{\ell}$ as
\begin{equation}
\rho_{A}^{\mathrm{NTU}} ( x ) = 
\left\{
\begin{array}{ll} 0 & \mbox{if~} x + A = 0 \\ k & \mbox{if~} x + A \in D_k, \, 0 \leq k \leq \ell - 1. \end{array}
\right.
\label{equation:NTU rho}
\end{equation}
Let $T_{A}^{\mathrm{NTU}} = ( T_{A, n}^{\mathrm{NTU}} )_{n \geq 0}$ be a sequence defined as
\begin{equation}
T_{A, n}^{\mathrm{NTU}} = \rho_{A}^{\mathrm{NTU}} ( R_n ), \quad n \geq 0,
\label{equation:NTU}
\end{equation}
and $t_{A}^{\mathrm{NTU}} = ( t_{A, n}^{\mathrm{NTU}} )_{n \geq 0}$ a sequence defined as
\begin{equation}
t_{A, n}^{\mathrm{NTU}} = \rho_{A}^{\mathrm{NTU}} ( g^n ), \quad n \geq 0.
\label{equation:Sid}
\end{equation}
$T_{A}^{\mathrm{NTU}}$ is called an $\ell$-ary generalized NTU sequence.

If $A = 0$, then $T_{0}^{\mathrm{NTU}}$ is an $\ell$-ary NTU sequence
\cite{Nogami-Tada-Uehara,Nasima-Nogami-Uehara-Moleros-Zaragoza}.
$T_{0}^{\mathrm{NTU}}$ is a periodic sequence of period $N_0 := \ell ( p^m - 1 ) / ( p - 1 )$.
The autocorrelation distribution of $T_{0}^{\mathrm{NTU}}$ have been explicitly obtained
in \cite[Property 6]{Nogami-Tada-Uehara} and \cite[Property 5]{Nasima-Nogami-Uehara-Moleros-Zaragoza}.
The linear complexity of $T_{0}^{\mathrm{NTU}}$ is $L ( T_{0}^{\mathrm{NTU}} ) = 2 ( p^m - 1 ) / ( p - 1 )$.
In particular, if $\ell = 2$, then $L ( T_{0}^{\mathrm{NTU}} )$ attains the maximal value $N_0$.

Assume that $A \neq 0$. Then $T_{A}^{\mathrm{NTU}}$ is a periodic sequence of period $N:= p^m - 1$.
The autocorrelation distribution of $T_{A}^{\mathrm{NTU}}$ have been explicitly obtained
in \cite[3.3]{Nogami-Uehara-Tsuchiya-Begum-Ino-Morelos-Zaragoza}.
In particular, if $\ell = 2$, then the autocorrelation distribution of $T_{A}^{\mathrm{NTU}}$ is
\begin{equation*}
R_{T_{A}^{\mathrm{NTU}}} ( \tau ) =
\left\{ \begin{array}{ll}
N & \mbox{if~} \tau = 0 \\ N_1^{( j )} & \mbox{if~} \tau = j \nu, \, 1 \leq j \leq p - 2 \\ N_2 & \mbox{otherwise}
\end{array} \right.
\end{equation*}
for $\tau \in \{ 0, 1, \dots , N - 1 \}$.
Here $N_1^{( j )}$ and $N_2$ are defined as
\begin{eqnarray}
\label{eqnarray:def of N_1}
N_1^{( j )} & = & p^{m - 1} \left\{ ( - 1 )^{\rho_{A}^{\mathrm{NTU}} ( - g^j A )}
+ ( - 1 )^{\rho_{A}^{\mathrm{NTU}} ( - g^{- j} A )} + ( - 1 )^{( j + 1 )} \right\} - 1, \\
N_2 & = & p^{m - 2} - 1,
\label{eqnarray:def of N_2}
\end{eqnarray}
respectively.

\begin{remark}
If $A = 1$, then $t_{1}^{\mathrm{NTU}}$ is an $\ell$-ary Sidel'nikov sequence \cite{Sidel'nikov,Lempel-Cohn-Eastman}.
\label{remark:Sidel'nikov seq}
\end{remark}

\begin{lemma}
Let $A_1, A_2 \in \mathbb{F}_p \setminus \{ 0 \}$.
If $A_1$ and $A_2$ belong to the same cyclotomic class,
then there exists $i \in \{ 0, 1, \dots ,p^m - 2 \}$ such that
$T_{A_2, n}^{\mathrm{NTU}} = T_{A_1, n + i}^{\mathrm{NTU}}, \, n \geq 0$.
\label{lemma:equivalence}
\end{lemma}

\begin{proof}
The statement follows from the shift-and-add property of m-sequences.
\end{proof}

\section{Linear Complexity of Generalized NTU Sequences}
\label{section:Linear Complexity of Generalized NTU Sequences}

In this section, we consider linear complexity of generalized NTU sequences.

Let $p, m, \omega, R = ( R_n )_{n \geq 0}, \nu, g, \ell$
and $D_j, \, 0 \leq j \leq \ell - 1$ be as in Sect. \ref{section:Generalized NTU Sequences}.
Let $A \in \mathbb{F}_p \setminus \{ 0 \}$.
Let $\rho_{A}^{\mathrm{NTU}} : \mathbb{F}_p \rightarrow \mathbb{F}_{\ell}$ be a map defined as (\ref{equation:NTU rho}),
$T_{A}^{\mathrm{NTU}} = ( T_{A, n}^{\mathrm{NTU}} )_{n \geq 0}$ a sequence defined as (\ref{equation:NTU}) and
$t_{A}^{\mathrm{NTU}} = ( t_{A, n}^{\mathrm{NTU}} )_{n \geq 0}$ a sequence defined as (\ref{equation:Sid}).

\subsection{The Case of $A \in D_0$}

Assume that $A \in D_0$.
Since $\rho_{A}^{\mathrm{NTU}} ( 0 ) = 0$, we can apply (\ref{equation:Chan-Games}) to $L ( T_{A}^{\mathrm{NTU}} )$.
By Lemma \ref{lemma:equivalence}, we have
\begin{equation*}
L ( T_{A}^{\mathrm{NTU}} ) = L ( T_{1}^{\mathrm{NTU}} ) = \nu L ( t_{1}^{\mathrm{NTU}} ).
\end{equation*}
Since $t_{1}^{\mathrm{NTU}}$ is a Sidel'nikov sequence
(see \cite{Sidel'nikov,Lempel-Cohn-Eastman} and Remark \ref{remark:Sidel'nikov seq}),
$L ( T_{A}^{\mathrm{NTU}} )$ is obtained from some known results
\cite{Helleseth-Yang,Kyureghyan-Pott,Meidl-Winterhof,Alaca-Millar,Zhang-Yang,Brandstatter-Meidl}
on $L ( t_{1}^{\mathrm{NTU}} )$.

\subsection{The Case of $A \in D_k, \, k \neq 0$}

Assume that $A \in D_k, \, k \neq 0$.
Since $\rho_{A}^{\mathrm{NTU}} ( 0 ) \neq 0$,
we can not apply (\ref{equation:Chan-Games}) to $L ( T_{A}^{\mathrm{NTU}} )$.
In order to apply (\ref{equation:Chan-Games}), we introduce a sequence whose any term is added by a complement of $k$.

We define a map $\overline{\rho_{A}^{\mathrm{NTU}}} : \mathbb{F}_p \rightarrow \mathbb{F}_{\ell}$ as
\begin{equation}
\overline{\rho_{A}^{\mathrm{NTU}}} ( x ) = \rho_{A}^{\mathrm{NTU}} ( x ) - k, \quad x \in \mathbb{F}_p.
\label{equation:Comp NTU rho}
\end{equation}
Let $\overline{T_{A}^{\mathrm{NTU}}} = ( \overline{T_{A, n}^{\mathrm{NTU}}} )_{n \geq 0}$ be a sequence defined as
\begin{equation}
\overline{T_{A, n}^{\mathrm{NTU}}} = \overline{\rho_{A}^{\mathrm{NTU}}} ( R_n ), \quad n \geq 0,
\label{equation:Comp NTU}
\end{equation}
and $\overline{t_{A}^{\mathrm{NTU}}} = ( \overline{t_{A, n}^{\mathrm{NTU}}} )_{n \geq 0}$ a sequence defined as
\begin{equation}
\overline{t_{A, n}^{\mathrm{NTU}}}= \overline{\rho_{A}^{\mathrm{NTU}}} ( g^n ), \quad n \geq 0.
\label{equation:Comp Sid}
\end{equation}
Since $\overline{\rho_{A}^{\mathrm{NTU}}} ( 0 ) = \rho_{A}^{\mathrm{NTU}} ( 0 ) - k = 0$,
we can apply (\ref{equation:Chan-Games}) to $L ( \overline{T_{A, n}^{\mathrm{NTU}}} )$.

\begin{lemma}[\cite{Lidl-Niederreiter} Theorem 8.62]
Let $T = ( T_n )_{n \geq 0}$ be a periodic sequence of period $N$ over $\mathbb{F}_{\ell}$.
For $a \in \mathbb{F}_{\ell} \setminus \{ 0 \}$,
let $\overline{T} = ( \overline{T}_n )_{n \geq 0}$ be a sequence defined by $\overline{T}_n = T_n + a, \, n \geq 0$.
\begin{enumerate}
\item If $\mathrm{ml}_1 \, ( x^N - 1 ) \leq \mathrm{ml}_1 \, ( T ( x ) )$, then $L ( \overline{T} ) = L ( T ) + 1$.
\item If $\mathrm{ml}_1 \, ( x^N - 1 ) = \mathrm{ml}_1 \, ( T ( x ) ) + 1$, then $L ( \overline{T} ) = L ( T ) - 1$.
\item If $\mathrm{ml}_1 \, ( x^N - 1 ) > \mathrm{ml}_1 \, ( T ( x ) ) + 1$, then $L ( \overline{T} ) = L ( T )$.
\end{enumerate}
\label{lemma:mult vs LC}
\end{lemma}

\begin{proposition}
Assume that $A \in D_k, \, k \neq 0$. If $\ell = 2$ and $p \equiv 1 \mathrm{~mod~} 4$ or $\ell \geq 3$, then
\begin{equation*}
L ( T_{A}^{\mathrm{NTU}} ) = \nu L ( t_{A}^{\mathrm{NTU}} ).
\end{equation*}
\end{proposition}

\begin{proof}
Since $p \equiv 1 \mathrm{~mod~} \ell$, $p^m - 1 \equiv p - 1 \equiv 0 \mathrm{~mod~} \ell$.
Hence $x^{p^m - 1} - 1 = ( x^{( p^m - 1 ) / \ell} - 1 )^{\ell}$ and
$x^{p - 1} - 1 = ( x^{( p - 1 ) / \ell} - 1 )^{\ell}$.
Therefore, $\mathrm{ml}_1 \, ( x^{p^m - 1} - 1 ) \geq \ell \geq 2$ and $\mathrm{ml}_1 \, ( x^{p - 1} - 1 ) \geq \ell \geq 2$.
On the other hand,
\begin{eqnarray*}
T_{A}^{\mathrm{NTU}} \, ( 1 ) & = &
\sum_{i = 0}^{p^m - 2} T_{A, i}^{\mathrm{NTU}} = \frac{( \ell - 1 ) ( p - 1 ) p^{m - 1}}{2} - k \\ & = &
\left\{ \begin{array}{ll}
0 & \mbox{if~} \ell = 2 \mbox{~and~} p \equiv 3 \mathrm{~mod~} 4 \\
\ell - k & \mbox{otherwise},
\end{array} \right.
\end{eqnarray*}
and
\begin{eqnarray*}
t_{A}^{\mathrm{NTU}} \, ( 1 ) & = &
\sum_{i = 0}^{p - 2} t_{A, i}^{\mathrm{NTU}} = \frac{( \ell - 1 ) ( p - 1 )}{2} - k \\ & = &
\left\{ \begin{array}{ll}
0 & \mbox{if~} \ell = 2 \mbox{~and~} p \equiv 3 \mathrm{~mod~} 4 \\
\ell - k & \mbox{otherwise}.
\end{array} \right.
\end{eqnarray*}
Hence the statement follows from Lemma \ref{lemma:Chan-Games} and Lemma \ref{lemma:mult vs LC}.
\end{proof}

\begin{example}
Let $p = 29, m = 3$ and $\ell = 2$.
%Let $\omega$ be a root of the primitive polynomial $x^3 + x^2 + 15 x + 8 \in \mathbb{F}_p [ x ]$.
Assume that $( A / p ) = - 1$. Then
\begin{eqnarray*}
&& L ( T_{A}^{\mathrm{NTU}} ) = L ( \overline{T_{A}^{\mathrm{NTU}}} ) = 24388, \\
&& L ( t_{A}^{\mathrm{NTU}} ) = L ( \overline{t_{A}^{\mathrm{NTU}}} ) = 28, \\
&& \nu = \frac{p^m - 1}{p - 1} = 871.
\end{eqnarray*}
\end{example}

If $\ell = 2, ( A / p ) = - 1$ and $p \equiv 3 \mathrm{~mod~} 4$,
then $\mathrm{ml}_1 \, ( T_{A}^{\mathrm{NTU}} \, ( x ) ) \geq 1$ and
$\mathrm{ml}_1 \, ( t_{A}^{\mathrm{NTU}} \, ( x ) ) \geq 1$.
Therefore, we evaluate the first Hasse derivative of $T_{A}^{\mathrm{NTU}} \, ( x )$ at $1$.

\begin{lemma}
Assume that $\ell = 2, ( A / p ) = - 1, p \equiv 3 \mathrm{~mod~} 4$ and $m \equiv 1 \mathrm{~mod~} 2$. Then
\begin{equation*}
{T_{A}^{\mathrm{NTU}}}^{( 1 )} ( 1 ) =
\left\{ \begin{array}{ll}
0 & \mbox{if~} p \equiv 3 \mathrm{~mod~} 8 \\
1 & \mbox{if~} p \equiv 7 \mathrm{~mod~} 8.
\end{array} \right.
\end{equation*}
\label{lemma:Hasse derivative}
\end{lemma}

\begin{proof}
Since
\begin{equation*}
{T_{A}^{\mathrm{NTU}}}^{( 1 )} ( x )
= \sum_{i = 1}^{p^m - 2} \binom{i}{1} T_{A, i}^{\mathrm{NTU}} x^{i - 1}
= \sum_{i = 0}^{( p^m - 3 ) / 2} T_{A, 2 i + 1}^{\mathrm{NTU}} x^{2 i},
\end{equation*}
we have
\begin{equation*}
{T_{A}^{\mathrm{NTU}}}^{( 1 )} ( 1 ) = \sum_{i = 0}^{( p^m - 3 ) / 2} T_{A, 2 i + 1}^{\mathrm{NTU}}.
\end{equation*}
For $0 \leq j \leq \nu - 1$, put $\boldsymbol{R}_j = ( R_{i \nu + j} )_{0 \leq i \leq p - 2} \in \mathbb{F}_p^{p - 1}$.

Assume that $R_j = 0$.
Since all the terms of $( T_{A, i \nu + j}^{\mathrm{NTU}} )_{0 \leq i \leq p - 2} \in \mathbb{F}_2^{p - 1}$ are $1$,
the number of terms satisfying $T_{A, i \nu + j}^{\mathrm{NTU}} = 1$ and $i \nu + j \equiv 1 \mathrm{~mod~} 2$
is $( p - 1 ) / 2$.

Assume that $R_j \neq 0$. Then
\begin{equation*}
\left\{ R_{2 i \nu + j} \, \middle| \, 0 \leq i \leq \frac{p - 3}{2} \right\} = D_k \quad \mbox{and} \quad
\left\{ R_{( 2 i + 1 ) \nu + j} \, \middle| \, 0 \leq i \leq \frac{p - 3}{2} \right\} = D_{1 - k},
\end{equation*}
where $k = 0$ or $1$.
Hence the number of terms of $( T_{A, i \nu + j}^{\mathrm{NTU}} )_{0 \leq i \leq p - 2} \in \mathbb{F}_2^{p - 1}$
such that $T_{A, i \nu + j}^{\mathrm{NTU}} = 1$ and $i \nu + j \equiv 1 \mathrm{~mod~} 2$
is $d ( 0, 1 ; A )$ or $d ( 1, 1 ; A )$, that is $( p - 3 ) / 4$ by Table \ref{table:diff param}.

On the other hand,
the number of terms of $( R_j )_{0 \leq j \leq \nu - 1} \in \mathbb{F}_p^{\nu}$ such that $R_j = 0$
is $( p^{m - 1} - 1 ) / ( p - 1 )$ (cf. \cite[Lemma 3]{Chan-Games}). Hence
\begin{eqnarray*}
\sum_{i = 0}^{( p^m - 3 ) / 2} T_{A, 2 i + 1}^{\mathrm{NTU}}
& = & \frac{p^{m - 1} - 1}{p - 1} \times \frac{p - 1}{2}
+ \left\{ \nu - \frac{p^{m - 1} - 1}{p - 1} \right\} \times \frac{p - 3}{4} \\
& = & \frac{p^{m - 1} ( p - 1 ) - 2}{4} \\
& = & \left\{ \begin{array}{ll}
0 & \mbox{if~} p \equiv 3 \mathrm{~mod~} 8 \\
1 & \mbox{if~} p \equiv 7 \mathrm{~mod~} 8.
\end{array} \right.
\end{eqnarray*}
\end{proof}

\begin{theorem}
Assume that $\ell = 2, ( A / p ) = - 1, p \equiv 3 \mathrm{~mod~} 4$ and $m \equiv 1 \mathrm{~mod~} 2$. Then
\begin{equation*}
L ( T_{A}^{\mathrm{NTU}} ) =
\left\{
\begin{array}{ll}
\nu \left( L ( t_{A}^{\mathrm{NTU}} ) + 1 \right) - 1 & \mbox{if~} p \equiv 3 \mathrm{~mod~} 8 \\
\nu \left( L ( t_{A}^{\mathrm{NTU}} ) - 1 \right) + 1 & \mbox{if~} p \equiv 7 \mathrm{~mod~} 8.
\end{array}
\right.
\end{equation*}
\end{theorem}

\begin{proof}
Since $p \equiv 3 \mathrm{~mod~} 4$ and $m \equiv 1 \mathrm{~mod~} 2$, $\mathrm{ml}_1 \, ( x^{p^m - 1} - 1 ) = 2$.
If $p \equiv 3 \mathrm{~mod~} 8$ (resp. if $p \equiv 7 \mathrm{~mod~} 8$),
then $\mathrm{ml}_1 \, ( T_{A}^{\mathrm{NTU}} \, ( x ) ) \geq 2$
(resp. $\mathrm{ml}_1 \, ( T_{A}^{\mathrm{NTU}} \, ( x ) ) = 1$) by Lemma \ref{lemma:Hasse derivative}. Hence, we have
\begin{equation*}
L ( T_{A}^{\mathrm{NTU}} )=
\left\{ \begin{array}{ll}
\nu L ( \overline{t_{A}^{\mathrm{NTU}}} ) - 1 & \mbox{if~} p \equiv 3 \mathrm{~mod~} 8 \\
\nu L ( \overline{t_{A}^{\mathrm{NTU}}} ) + 1 & \mbox{if~} p \equiv 7 \mathrm{~mod~} 8
\end{array} \right.
\end{equation*}
by Lemma \ref{lemma:Chan-Games} and Lemma \ref{lemma:mult vs LC}.
On the other hand, since $p \equiv 3 \mathrm{~mod~} 4$, $\mathrm{ml}_1 \, ( x^{p - 1} - 1 ) = 2$.
By Table \ref{table:diff param},
\begin{eqnarray*}
{t_{A}^{\mathrm{NTU}}}^{( 1 )} ( 1 )
& = & \sum_{i = 0}^{( p - 3 ) / 2} t_{A, 2 i + 1}^{\mathrm{NTU}}
= d ( 1, 1 ; A ) = \frac{p - 3}{4} \\
& = & \left\{\begin{array}{ll}
0 & \mbox{if~} p \equiv 3 \mathrm{~mod~} 8 \\
1 & \mbox{if~} p \equiv 7 \mathrm{~mod~} 8.
\end{array} \right.
\end{eqnarray*}
Hence, if $p \equiv 3 \mathrm{~mod~} 8$ (resp. if $p \equiv 7 \mathrm{~mod~} 8$),
then $\mathrm{ml}_1 \, ( t_{A}^{\mathrm{NTU}} \, ( x ) ) \geq 2$
(resp. $\mathrm{ml}_1 \, ( t_{A}^{\mathrm{NTU}} \, ( x ) ) = 1$).
Therefore the statement follows from Lemma \ref{lemma:mult vs LC}.
\end{proof}

\begin{example}
Let $p = 43, m = 3$ and $\ell = 2$.
%Let $\omega$ be a root of the primitive polynomial $x^3 + 6 x^2 + 16 x + 14 \in \mathbb{F}_p [ x ]$.
Assume that $( A / p ) = - 1$. Then
\begin{eqnarray*}
&& L ( T_{A}^{\mathrm{NTU}} ) = 77612, \, L ( \overline{T_{A}^{\mathrm{NTU}}} ) = 77613, \\
&& L ( t_{A}^{\mathrm{NTU}} ) = 40, \, L ( \overline{t_{A}^{\mathrm{NTU}}} ) = 41, \\
&& \nu = \frac{p^m - 1}{p - 1} = 1893.
\end{eqnarray*}
\end{example}

\begin{example}
Let $p = 47, m = 3$ and $\ell = 2$.
%Let $\omega$ be a root of the primitive polynomial $x^3 + 23 x^2 + 25 x + 9 \in \mathbb{F}_p [ x ]$.
Assume that $( A / p ) = - 1$. Then
\begin{eqnarray*}
&& L ( T_{A}^{\mathrm{NTU}} ) = 99309, \, L ( \overline{T_{A}^{\mathrm{NTU}}} ) = 99308, \\
&& L ( t_{A}^{\mathrm{NTU}} ) = 45, \, L ( \overline{t_{A}^{\mathrm{NTU}}} ) = 44, \\
&& \nu = \frac{p^m - 1}{p - 1} = 2257.
\end{eqnarray*}
\end{example}

\begin{remark}
If $m \equiv 0 \mathrm{~mod~} 2$,
we need to evaluate higher Hasse derivative of $T_{A}^{\mathrm{NTU}} \, ( x )$ at $1$.
In \cite{Heguri-Nogami-Uehara-Tsuchiya}, a formula was conjectured by numerical experiments as follows:
Assume that $\ell = 2, ( A / p ) = - 1, p \equiv 3 \mathrm{~mod~} 4$ and $m \equiv 0 \mathrm{~mod~} 2$. Then
\begin{equation*}
L ( T_{A}^{\mathrm{NTU}} ) =
\left\{
\begin{array}{ll}
\nu \left( L ( t_{A}^{\mathrm{NTU}} ) + 1 \right) & \mbox{if~} p \equiv 3 \mathrm{~mod~} 8 \\
\nu \left( L ( t_{A}^{\mathrm{NTU}} ) - 1 \right) + 1 & \mbox{if~} p \equiv 7 \mathrm{~mod~} 8.
\end{array}
\right.
\end{equation*}
\end{remark}

\subsection{Conditions to Have a Large Linear Complexity}

In order to show conditions to have a large linear complexity, we refer to the result shown by
Meidl and Winterhof \cite{Meidl-Winterhof} and Brandst\"atter and Meidl \cite{Brandstatter-Meidl}.

\begin{lemma}[\cite{Meidl-Winterhof} Proposition 3, \cite{Brandstatter-Meidl} Proposition 1]
Let $r \neq \ell$ be a prime divisor of $p - 1$.
If $\ell$ is a primitive root in $\mathbb{F}_r$ and $r \geq \sqrt{p} + 1$,
then for each $r$th root of unity $\beta \neq 1$ we have $t_{A}^{\mathrm{NTU}} ( \beta ) \neq 0$.
\label{lemma:Meidl-Winterhof}
\end{lemma}

Lemma \ref{lemma:Meidl-Winterhof} was shown in the case of $A = 1$.
However, Lemma \ref{lemma:Meidl-Winterhof} holds for any $A \in \mathbb{F}_p \setminus \{ 0 \}$.
Therefore, we obtain conditions to have a large linear complexity in the case of $\ell = 2$ as follows:

\begin{corollary}
Assume that $\ell = 2$. Let $p$ be a prime number of the form $p = 2^s r + 1$,
where $r$ is an odd prime number such that $2$ is a primitive root in $\mathbb{F}_r$ and $r \geq \sqrt{p} + 1$.
\begin{enumerate}
\item
If $A = 1$ and $s = 1$ or $( A / p ) = - 1$ and $s \geq 2$,
then the minimal polynomials $m_{t_{A}^{\mathrm{NTU}}} ( x )$ and $m_{T_{A}^{\mathrm{NTU}}} ( x )$
of $t_{A}^{\mathrm{NTU}}$ and $T_{A}^{\mathrm{NTU}}$ are given as
\begin{equation*}
m_{t_{A}^{\mathrm{NTU}}} ( x ) = x^{p - 1} + 1, \, m_{T_{A}^{\mathrm{NTU}}} ( x ) = x^{p^m - 1} + 1,
\end{equation*}
respectively. Therefore the linear complexity $L ( t_{A}^{\mathrm{NTU}} )$ and $L ( T_{A}^{\mathrm{NTU}} )$
of $t_{A}^{\mathrm{NTU}}$ and $T_{A}^{\mathrm{NTU}}$ are given as
\begin{equation*}
L ( t_{A}^{\mathrm{NTU}} ) = p - 1, \, L ( T_{A}^{\mathrm{NTU}} ) = p^m - 1,
\end{equation*}
respectively.
\item
If $( A / p ) = - 1, m \equiv 1 \mathrm{~mod~} 2$ and $p \equiv 7 \mathrm{~mod~} 8$,
then the minimal polynomials $m_{t_{A}^{\mathrm{NTU}}} ( x )$ and $m_{T_{A}^{\mathrm{NTU}}} ( x )$
of $t_{A}^{\mathrm{NTU}}$ and $T_{A}^{\mathrm{NTU}}$ are given as
\begin{equation*}
m_{t_{A}^{\mathrm{NTU}}} ( x ) = \frac{x^{p - 1} + 1}{x + 1}, \,
m_{T_{A}^{\mathrm{NTU}}} ( x ) = \frac{( x + 1 ) ( x^{p^m - 1} + 1 )}{x^{2 \nu} + 1},
\end{equation*}
respectively. Therefore the linear complexity $L ( t_{A}^{\mathrm{NTU}} )$ and $L ( T_{A}^{\mathrm{NTU}} )$
of $t_{A}^{\mathrm{NTU}}$ and $T_{A}^{\mathrm{NTU}}$ are given as
\begin{equation*}
L ( t_{A}^{\mathrm{NTU}} ) = p - 2, \, L ( T_{A}^{\mathrm{NTU}} ) = \frac{p^{m + 1} - 3 p^m + 2}{p - 1},
\end{equation*}
respectively.
\end{enumerate}
\label{corollary:large LC}
\end{corollary}

\begin{proof}
We only show the case that $( A / p ) = - 1, m \equiv 1 \mathrm{~mod~} 2$ and $p \equiv 7 \mathrm{~mod~} 8$
(the proofs for the other cases are trivial).
By the argument in the previous subsection,
\begin{equation*}
m_{t_{A}^{\mathrm{NTU}}} ( x ) = ( x + 1 ) m_{\overline{t_{A}^{\mathrm{NTU}}}} ( x ), \,
m_{T_{A}^{\mathrm{NTU}}} ( x ) = ( x + 1 ) m_{\overline{T_{A}^{\mathrm{NTU}}}} ( x ).
\end{equation*}
Noting that $m_{\overline{T_{A}^{\mathrm{NTU}}}} ( x ) = m_{\overline{t_{A}^{\mathrm{NTU}}}} ( x^{\nu} )$
(cf. \cite[p.551, Theorem]{Chan-Games}), we have
\begin{equation*}
m_{T_{A}^{\mathrm{NTU}}} ( x ) = \frac{x + 1}{x^{\nu} + 1} m_{t_{A}^{\mathrm{NTU}}} ( x^{\nu} )
= \frac{( x + 1 ) ( x^{p^m - 1} + 1 )}{x^{2 \nu} + 1}.
\end{equation*}
\end{proof}

\begin{remark}
For the case of $\ell = 3$,
we can obtain some conditions to have a large linear complexity by \cite[Theorem 1]{Brandstatter-Meidl}.
\end{remark}

\section{Interleaved Sequences of Generalized NTU Sequences}
\label{section:Interleaved Sequences of Generalized NTU Sequences}

Generalized NTU sequences have good properties for linear complexity as discussed in the previous section,
however they do not have the balance property.
In order to construct sequences that have the balance property,
we introduce interleaved sequences as in the case of NTU sequences \cite{Tsuchiya-Nogami-Uehara}.
Gong \cite{Gong1995} introduced the concept of interleaved sequences that were obtained by merging sequences, and
sequence interleaving was widely used to improve the balance property of existing sequences
such as in \cite{Tang-Ding,Gu-Hwang-Han-Kim-Jin,Gong2002,Edemskiy,Chung-Yang,Zeng-Zeng-Zhang-Xuan,Yu-Gong}.

In this section, we propose interleaved sequences of a generalized NTU sequence and its complement.
By interleaving them, the interleaved sequence has the balance property by the interleaved structure.
Furthermore, we show periodic autocorrelation and linear complexity of the interleaved sequences.

Let $p, m, \omega, R = ( R_n )_{n \geq 0}, \nu, g, \ell$
and $D_j, \, 0 \leq j \leq \ell - 1$ be as in Sect. \ref{section:Generalized NTU Sequences}.
Let $A \in \mathbb{F}_p$.
Let $\rho_{A}^{\mathrm{NTU}} : \mathbb{F}_p \rightarrow \mathbb{F}_{\ell}$ be a map defined as (\ref{equation:NTU rho}) and
$T_{A}^{\mathrm{NTU}} = ( T_{A, n}^{\mathrm{NTU}} )_{n \geq 0}$ a sequence defined as (\ref{equation:NTU}).
Let $\overline{\rho_{A}^{\mathrm{NTU}}} : \mathbb{F}_p \rightarrow \mathbb{F}_{\ell}$
be a map defined as (\ref{equation:Comp NTU rho}) and
$\overline{T_{A}^{\mathrm{NTU}}} = ( \overline{T_{A, n}^{\mathrm{NTU}}} )_{n \geq 0}$
a sequence defined as (\ref{equation:Comp NTU}).
Throughout this section, assume that $\ell = 2$.
Put $N = p^m - 1$.

In the case of $A = 0$, interleaved sequences of NTU sequences were studied in \cite{Tsuchiya-Nogami-Uehara}.
We apply a similar argument to the case of $A \neq 0$.

\subsection{Interleaved Sequences of Generalized NTU Sequences}

For $e \in \{ 0, 1, \dots , N - 1 \}$,
we define a sequence $S^e = ( S^e_n )_{ n \geq 0 }$ as the interleaved sequence of
$\left\{ T_{A}^{\mathrm{NTU}}, L^{e} \left( \overline{T_{A}^{\mathrm{NTU}}} \right) \right\}$.

\begin{example}
Let $p = 3$ and $m = 2$, so that $N = 8$.
Let $\omega$ be a root of the primitive polynomial $x^2 + 2 x + 2$.
Then $T_{A}^{\mathrm{NTU}}$ is given as
\begin{equation*}
T_{A}^{\mathrm{NTU}} = ( 0, 1, 0, 1, 1, 0, 0, 0, \dots ).
\end{equation*}
Put $e = 2$. Then $S^e = S^2$ is given as
\begin{equation*}
S^2 = ( 0, 1, 1, 0, 0, 0, 1, 1, 1, 1, 0, 1, 0, 1, 0, 0, \dots ).
\end{equation*}
\end{example}

\subsection{Periodic Autocorrelation of the Proposed Sequences}

We show some lemmas to obtain the periodic autocorrelation of the proposed sequences.

\begin{lemma}
For $\tau \in \{ 0, 1, \dots , N - 1 \}$,
\begin{equation*}
R_{T_{A}^{\mathrm{NTU}}, \overline{T_{A}^{\mathrm{NTU}}}} ( \tau ) = - R_{T_{A}^{\mathrm{NTU}}} ( \tau ).
\end{equation*}
\label{lemma:CC of T and T^-}
\end{lemma}

\begin{proof}
It follows that
\begin{eqnarray*}
R_{T_{A}^{\mathrm{NTU}}, \overline{T_{A}^{\mathrm{NTU}}}} ( \tau )
& = & \sum_{i = 0}^{N - 1} ( - 1 )^{T_{A, i}^{\mathrm{NTU}} + \overline{T_{A, i + \tau}^{\mathrm{NTU}}}}
= \sum_{i = 0}^{N - 1} ( - 1 )^{T_{A, i}^{\mathrm{NTU}} + T_{A, i + \tau}^{\mathrm{NTU}} + 1} \\
& = & \# \left\{ i \mid T_{A, i}^{\mathrm{NTU}} = T_{A, i + \tau}^{\mathrm{NTU}} + 1, \, 0 \leq i \leq N - 1 \right\} \\
& & - \# \left\{ i \mid T_{A, i}^{\mathrm{NTU}} \neq T_{A, i + \tau}^{\mathrm{NTU}} + 1, \, 0 \leq i \leq N - 1 \right\}  \\
& = & \# \left\{ i \mid T_{A, i}^{\mathrm{NTU}} \neq T_{A, i + \tau}^{\mathrm{NTU}}, \, 0 \leq i \leq N - 1 \right\} \\
& & - \# \left\{ i \mid T_{A, i}^{\mathrm{NTU}} = T_{A, i + \tau}^{\mathrm{NTU}}, \, 0 \leq i \leq N - 1 \right\} \\
& = & - R_{T_{A}^{\mathrm{NTU}}} ( \tau ).
\end{eqnarray*}
\end{proof}

\begin{lemma}
If $\tau = 2 \tau_0, \tau_0 \in \{ 0, 1, \dots , N - 1 \}$, then
\begin{equation*}
R_{S^{e_1}, S^{e_2}} ( \tau ) = R_{T_{A}^{\mathrm{NTU}}} ( \tau_0 ) + R_{T_{A}^{\mathrm{NTU}} } ( e_2 - e_1 + \tau_0 ).
\end{equation*}
If $\tau = 2 \tau_0 + 1, \tau_0 \in \{ 0, 1, \dots , N - 1 \}$, then
\begin{equation*}
R_{S^{e_1}, S^{e_2}} ( \tau ) = - R_{T_{A}^{\mathrm{NTU}}} ( e_2 + \tau_0 ) - R_{T_{A}^{\mathrm{NTU}}} ( e_1 - \tau_0 - 1 ).
\end{equation*}
\label{lemma:description of R(tau) by T}
\end{lemma}

\begin{proof}
Assume that $\tau = 2 \tau_0, \tau_0 \in \{ 0, 1, \dots , N - 1 \}$.
Then
\begin{eqnarray*}
R_{ S^{e_1}, S^{e_2} } ( \tau )
& = & R_{T_{A}^{\mathrm{NTU}}} ( \tau_0 ) + R_{\overline{T_{A}^{\mathrm{NTU}}}} ( e_2 - e_1 + \tau_0 ) \\
& = & R_{T_{A}^{\mathrm{NTU}}} ( \tau_0 ) + R_{T_{A}^{\mathrm{NTU}}} ( e_2 - e_1 + \tau_0 ).
\end{eqnarray*}

Assume that $\tau = 2 \tau_0 + 1, \tau_0 \in \{ 0, 1, \dots , N - 1 \}$.
By Lemma \ref{lemma:CC of T and T^-},
\begin{eqnarray*}
R_{ S^{e_1}, S^{e_2} } ( \tau )
& = & R_{T_{A}^{\mathrm{NTU}}, \overline{T_{A}^{\mathrm{NTU}}}} ( e_2 + \tau_0 )
+ R_{\overline{T_{A}^{\mathrm{NTU}}}, T_{A}^{\mathrm{NTU}}} ( \tau_0 + 1 - e_1 ) \\
& = & R_{T_{A}^{\mathrm{NTU}}, \overline{T_{A}^{\mathrm{NTU}}}} ( e_2 + \tau_0 )
+ R_{T_{A}^{\mathrm{NTU}}, \overline{T_{A}^{\mathrm{NTU}}}} ( e_1 - \tau_0 - 1 ) \\
& = & - R_{T_{A}^{\mathrm{NTU}}} ( e_2 + \tau_0 ) - R_{T_{A}^{\mathrm{NTU}}} ( e_1 - \tau_0 - 1 ).
\end{eqnarray*}
\end{proof}

The periodic autocorrelation of the proposed sequence is obtained as follows:

\begin{theorem}
For $e \in \{ 0, 1, \dots , N - 1 \}$, let $S^e$ be the interleaved sequence of
$\left\{ T_{A}^{\mathrm{NTU}}, L^e \left( \overline{T_{A}^{\mathrm{NTU}}} \right) \right\}$.
For $\tau = 2 \tau_0, \tau_0 \in \{ 0, 1, \dots , N - 1 \}$,
the periodic autocorrelation $R_{S^e} ( \tau )$ of $S^e$ is given as
\begin{equation}
\left\{
\begin{array}{ll}
2 N & \mbox{if~} \tau_0 = 0 \\
2 N_1^{( j )} & \mbox{if~} \tau_0 = j \nu, \, 1 \leq j \leq p - 2 \\
2 N_2 & \mbox{otherwise}.
\end{array}
\right.
\label{equation:AC for even tau}
\end{equation}
For $\tau = 2 \tau_0 + 1, \tau_0 \in \{ 0, 1, \dots , N - 1 \}$,
the periodic autocorrelation $R_{S^e} ( \tau )$ of $S^e$ is given as
\begin{equation}
\left\{
\begin{array}{ll}
- N - N_2 & \mbox{if~} \tau_0 \equiv - e, e - 1 \mathrm{~mod~} N \\
- N_1^{( j )} - N_2 & \mbox{if~} \tau_0 \equiv - e + j \nu, e - 1 - j \nu \mathrm{~mod~} N, \\ & 1 \leq j \leq p - 2 \\
- 2 N_2 & \mbox{otherwise}
\end{array}
\right.
\label{equation:AC for odd tau 1}
\end{equation}
if $2 e \not\equiv 1 + j \nu \mathrm{~mod~} N$ for any $j \in \{ 1, \dots , p - 2 \}$, and
\begin{equation}
\left\{
\begin{array}{ll}
- N - N_1^{( j_0 )} & \mbox{if~} \tau_0 \equiv - e, e - 1 \mathrm{~mod~} N \\
- N_1^{( j )} - N_1^{( j_0 - j \mathrm{~mod~} p - 1 )}
& \mbox{if~} \tau_0 \equiv - e + j \nu, e - 1 - j \nu  \mathrm{~mod~} N, \\
& 1 \leq j \leq p - 2, \, j \neq j_0 \\
- 2 N_2 & \mbox{otherwise}
\end{array}
\right.
\label{equation:AC for odd tau 2}
\end{equation}
if there exists a $j_0 \in \{ 1, \dots , p - 2 \}$ such that $2 e \equiv 1 + j_0 \nu \mathrm{~mod~} N$.
Here $N_1^{( j )}$ and $N_2$ are defined as (\ref{eqnarray:def of N_1}) and (\ref{eqnarray:def of N_2}), respectively.
\label{theorem:AC}
\end{theorem}

\begin{proof}
Assume that $\tau = 2 \tau_0, \tau_0 \in \{ 0, 1, \dots , N - 1 \}$.
Since $R_{S^e} ( \tau ) = 2 R_{T_{A}^{\mathrm{NTU}}} ( \tau_0 )$
by Lemma \ref{lemma:description of R(tau) by T}, $R_{S^e} ( \tau )$ is given as (\ref{equation:AC for even tau}).

Assume that $\tau = 2 \tau_0 + 1, \tau_0 \in \{ 0, 1, \dots , N - 1 \}$.
Then we have $R_{S^e} ( \tau ) = - R_{T_{A}^{\mathrm{NTU}}} ( e + \tau_0 ) - R_{T_{A}^{\mathrm{NTU}}} ( e - \tau_0 - 1 )$
by Lemma \ref{lemma:description of R(tau) by T}.
On the other hand,
\begin{eqnarray*}
R_{T_{A}^{\mathrm{NTU}}} ( e + \tau_0 ) & = &
\left\{
\begin{array}{ll}
N & \mbox{if~} \tau_0 \equiv - e \mathrm{~mod~} N \\
N_1^{( j )} & \mbox{if~} \tau_0 \equiv - e + j \nu \mathrm{~mod~} N, 1 \leq j \leq p - 2 \\
N_2 & \mbox{otherwise}
\end{array}
\right.
\end{eqnarray*}
and
\begin{eqnarray*}
R_{T_{A}^{\mathrm{NTU}}} ( e - \tau_0 - 1 ) & = &
\left\{
\begin{array}{ll}
N & \mbox{if~} \tau_0 \equiv e - 1 \mathrm{~mod~} N \\
N_1^{( j )} & \mbox{if~} \tau_0 \equiv e - 1 - j \nu \mathrm{~mod~} N, 1 \leq j \leq p - 2 \\
N_2 & \mbox{otherwise}.
\end{array}
\right.
\end{eqnarray*}
If $2 e \not\equiv 1 + j \nu \mathrm{~mod~} N$ for any $j \in \{ 1, \dots , p - 2 \}$,
then $- e, - e + j \nu, e - 1, e - 1 - j \nu$ are incongruent modulo $N$.
Hence $R_{S^e} ( \tau )$ is given as (\ref{equation:AC for odd tau 1}).
If there exists a $j_0 \in \{ 1, \dots , p - 2 \}$ such that $2 e \equiv 1 + j_0 \nu \mathrm{~mod~} N$,
then $- e + j_0 \nu \equiv e - 1 \mathrm{~mod~} N$ and $- e \equiv e - 1 - j_0 \nu \mathrm{~mod~} N$.
Hence $R_{S^e} ( \tau )$ is given as (\ref{equation:AC for odd tau 2}).
\end{proof}

\begin{corollary}
$S^e$ has period $2 N = 2 ( p^m - 1 )$.
\end{corollary}

By the interleaved structure, $S^e$ has the balance property,
namely the number of zeros and ones in one period of $S^e$ are $N$.

\begin{remark}
If $m$ is even, then $2 e \not\equiv 1 + j \nu \mathrm{~mod~} N$ for any $j \in \{ 1, \dots , p - 2 \}$.
\end{remark}

\begin{example}
Let $p = 7$ and $m = 2$, so that $N = 48$.
%Let $\omega$ be a root of the primitive polynomial $x^2 + a x + b$.
Figure \ref{figure:AC in p = 7, m = 2, e = 2} describes the graph of the periodic autocorrelation of $S^2$.
\begin{figure}[tb]
\begin{center}
\includegraphics[scale = 0.55]{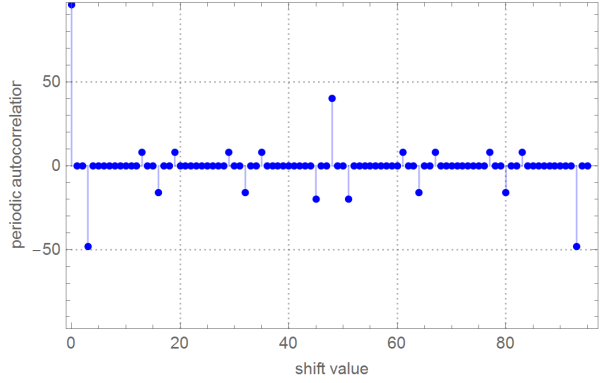}
\end{center}
\caption{The graph of the periodic autocorrelation of $S^2$ in the case of $p = 7, m = 2$.}
\label{figure:AC in p = 7, m = 2, e = 2}
\end{figure}
\end{example}

\begin{example}
Let $p = 5$ and $m = 3$, so that $N = 124$.
%Let $\omega$ be a root of the primitive polynomial $x^3 + a x^2 + b x + c$.
Figure \ref{figure:AC in p = 5, m = 3, e = 6} and Fig. \ref{figure:AC in p = 5, m = 3, e = 16}
describe the graphs of the periodic autocorrelation of $S^6$ and $S^{16}$, respectively.
Note that $\nu = 31$ and $2 \cdot 16 = 1 + 1 \cdot 31 = 32$.
\begin{figure}[tb]
\begin{center}
\includegraphics[scale = 0.55]{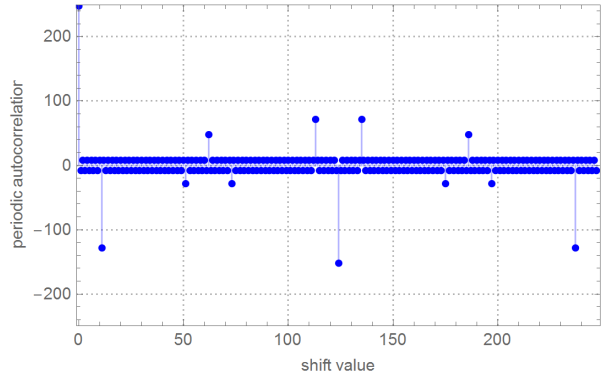}
\end{center}
\caption{The graph of the periodic autocorrelation of $S^6$ in the case of $p = 5, m = 3$.}
\label{figure:AC in p = 5, m = 3, e = 6}
\end{figure}
\begin{figure}[tb]
\begin{center}
\includegraphics[scale = 0.55]{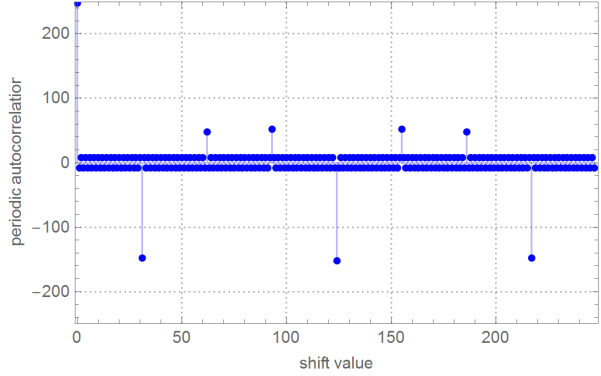}
\end{center}
\caption{The graph of the periodic autocorrelation of $S^{16}$ in the case of $p = 5, m = 3$.}
\label{figure:AC in p = 5, m = 3, e = 16}
\end{figure}
\end{example}

\subsection{Linear Complexity of the Proposed Sequences}

The minimal polynomial and linear complexity of the proposed sequence are obtained as follows:

\begin{theorem}
For $e \in \{ 0, 1, \dots , N - 1 \}$, let $S^e$ be the interleaved sequence of
$\left\{ T_{A}^{\mathrm{NTU}}, L^e \left( \overline{T_{A}^{\mathrm{NTU}}} \right) \right\}$.
Assume that $p$ is a prime number of the form $p = 2^s r + 1$,
where $r$ is an odd prime number such that $2$ is a primitive root in $\mathbb{F}_r$ and $r \geq \sqrt{p} + 1$.
\begin{enumerate}
\item
If $A = 1$ and $s = 1$ or $( A / p ) = - 1$ and $s \geq 2$,
then the minimal polynomial $m_{S^e} ( x )$ of $S^e$ is given as
\begin{equation*}
m_{S^e} ( x ) = \frac{x^{2 N} + 1}{x^{G ( N, e )} + 1},
\end{equation*}
where $G ( N, e ) = \mathrm{gcd} \, ( - 2 e + 1 ~\underline{\mathrm{mod}}~ N / 2^{\nu_2 ( N )}, N / 2^{\nu_2 ( N )} )$ and
$\nu_2 ( N )$ is the exponent of the largest power of $2$ which divides $N$.
Therefore the linear complexity $L ( S^e )$ of $S^e$ is given as
\begin{equation*}
L ( S^e ) = 2 N - G ( N, e ) = 2 L ( T_{A}^{\mathrm{NTU}} ) - G ( N, e ).
\end{equation*}
\item
If $( A / p ) = - 1, m \equiv 1 \mathrm{~mod~} 2$ and $p \equiv 7 \mathrm{~mod~} 8$,
then the minimal polynomial $m_{S^e} ( x )$ of $S^e$ is given as
\begin{equation*}
m_{S^e} ( x ) = \frac{( x^{2 N} + 1 ) ( x^2 + 1 ) ( x^{H_0 ( \nu, e )} + 1 )}{( x^{4 \nu} + 1 ) ( x^{H_1 ( N, e )} + 1 )},
\end{equation*}
where $H_0 ( \nu, e ) = \mathrm{gcd} \, ( - 2 e + 1 ~\underline{\mathrm{mod}}~ \nu, \nu )$ and
$H_1 ( N, e ) = \mathrm{gcd} \, ( - 2 e + 1 ~\underline{\mathrm{mod}}~ N / 2 , N / 2 )$.
Therefore the linear complexity $L ( S^e )$ of $S^e$ is given as
\begin{eqnarray*}
L ( S^e ) & = & 2 N + 2 - 4 \nu + H_0 ( \nu, e ) - H_1 ( N, e ) \\
& = & 2 L ( T_{A}^{\mathrm{NTU}} ) + H_0 ( \nu, e ) - H_1 ( N, e ).
\end{eqnarray*} 
\end{enumerate}
\label{theorem:LC of int seq}
\end{theorem}

\begin{proof}
Since
\begin{eqnarray*}
L^e \left( \overline{T_{A}^{\mathrm{NTU}}} \right) ( x )
& \equiv & x^{N - e} \overline{T_{A}^{\mathrm{NTU}}} ( x ) \mathrm{~mod~} x^N + 1 \\
& \equiv & x^{N - e} T_{A}^{\mathrm{NTU}} ( x ) + x^{N - e} \cdot \frac{x^N + 1}{x + 1} \mathrm{~mod~} x^N + 1,
\end{eqnarray*}
we have
\begin{eqnarray*}
S^e ( x )
& = & T_{A}^{\mathrm{NTU}} ( x^2 ) + x L^e \left( \overline{T_{A}^{\mathrm{NTU}}} \right) ( x^2 ) \\
& \equiv & T_{A}^{\mathrm{NTU}} ( x^2 ) + x \left\{ x^{2 N - 2 e} T_{A}^{\mathrm{NTU}} ( x^2 )
+ x^{2 N - 2 e} \cdot \frac{x^{2 N} + 1}{x^2 + 1} \right\} \mathrm{~mod~} x^{2 N} + 1 \\
& \equiv & ( x^{2 N - 2 e + 1} + 1 ) T_{A}^{\mathrm{NTU}} ( x^2 ) + x^{2 N - 2 e + 1} \cdot \frac{x^{2 N} + 1}{x^2 + 1}
\mathrm{~mod~} x^{2 N} + 1.
\end{eqnarray*}

Assume that $A = 1$ and $s = 1$ or $( A / p ) = - 1$ and $s \geq 2$.
By Corollary \ref{corollary:large LC}, $L ( T_{A}^{\mathrm{NTU}} ) = N$.
Since $2 N - 2 e + 1 \equiv 1 \mathrm{~mod~} 2$, $x^2 + 1$ does not divide $x^{2 N - 2 e + 1} + 1$.
Hence it follows that
\begin{eqnarray*}
\mathrm{gcd} \, ( x^{2 N} + 1, S^e ( x ) ) & = & \mathrm{gcd} \, ( x^{2 N} + 1, x^{2 N - 2 e + 1} + 1 ) \\
& = & x^{\mathrm{gcd} \, ( 2 N , 2 N - 2 e + 1 )} + 1 \\
& = & x^{G ( N, e )} + 1.
\end{eqnarray*}
Therefore we have
\begin{equation*}
m_{S^e} ( x ) = \frac{x^{2N} + 1}{\mathrm{gcd} \, ( x^{2 N} + 1, S^e ( x ) )} = \frac{x^{2 N} + 1}{x^{G ( N, e )} + 1},
\end{equation*}
and
\begin{eqnarray*}
L ( S^e ) & = & 2 N - \deg \mathrm{gcd} \, ( x^{2 N} + 1, S^e ( x ) ) \\
& = & 2 N - \deg \left( x^{G ( N, e )} + 1 \right) \\
& = & 2 N - G ( N, e ).
\end{eqnarray*}

Assume that $( A / p ) = - 1, m \equiv 1 \mathrm{~mod~} 2$ and $p \equiv 7 \mathrm{~mod~} 8$.
Put $\varphi ( x ) = \mathrm{gcd} \, ( x^N + 1, T_{A}^{\mathrm{NTU}} ( x ) ) = ( x^{2 \nu} + 1 ) / ( x + 1 )$ and
$\psi ( x ) = T_{A}^{\mathrm{NTU}} ( x ) / \varphi ( x )$.
By Corollary \ref{corollary:large LC},
$\mathrm{ml}_1 \, ( \varphi ( x ) ) = \mathrm{ml}_1 \, ( m_{T_{A}^{\mathrm{NTU}}} ( x ) ) = 1,
\mathrm{ml}_1 \, ( \psi ( x ) ) = 0$ and $\mathrm{gcd} \, ( m_{T_{A}^{\mathrm{NTU}}} ( x ), \psi ( x ) ) = 1$.
Then
\begin{eqnarray*}
& & S^e ( x ) \\
& \equiv & \varphi ( x^2 ) \left\{ ( x^{2 N - 2 e + 1} + 1 ) \psi ( x^2 )
+ x^{2 N - 2 e + 1} \cdot \frac{m_{T_{A}^{\mathrm{NTU}}} ( x^2 )}{x^2 + 1} \right\} \mathrm{~mod~} x^{2 N} + 1 \\
& \equiv & \varphi ( x^2 ) \left\{ ( x^{2 N - 2 e + 1} + 1 ) \psi ( x^2 )
+ x^{2 N - 2 e + 1} \cdot \frac{x^{2 N} + 1}{x^{4 \nu} + 1} \right\} \mathrm{~mod~} x^{2 N} + 1.
\end{eqnarray*}
Hence it follows that
\begin{eqnarray*}
\mathrm{gcd} \, ( x^{2 N} + 1, S^{e} (x) ) & = & \varphi ( x^2 ) \cdot
\frac{x^{\mathrm{gcd} \, ( 2 N - 2 e + 1, 2 N )} + 1}{x^{\mathrm{gcd} \, ( 2 N - 2 e + 1, 4 \nu )} + 1} \\
& = & \varphi ( x^2 ) \cdot \frac{x^{H_1 ( N, e )} + 1}{x^{H_0 ( \nu, e )} + 1 } \\
& = & \frac{( x^{4 \nu} + 1 ) ( x^{H_1 ( N, e )} + 1 )}{( x^2 + 1 ) ( x^{H_0 ( \nu, e )} + 1 )}.
\end{eqnarray*}
Therefore we have
\begin{equation*}
m_{S^e} ( x ) = \frac{x^{2N} + 1}{\mathrm{gcd} \, ( x^{2 N} + 1, S^e ( x ) )}
= \frac{( x^{2 N} + 1 ) ( x^2 + 1 ) ( x^{H_0 ( \nu, e )} + 1 )}{( x^{4 \nu} + 1 ) ( x^{H_1 ( N, e )} + 1 )},
\end{equation*}
and
\begin{eqnarray*}
L ( S^e ) & = & 2 N - \deg \mathrm{gcd} \, ( x^{2 N} + 1, S^e ( x ) ) \\
& = & 2 N + 2 - 4 \nu + H_0 ( \nu, e ) - H_1 ( N, e ) \\
& = & 2 L ( T_{A}^{\mathrm{NTU}} ) + H_0 ( \nu, e ) - H_1 ( N, e ).
\end{eqnarray*}
\end{proof}

\begin{corollary}
Assume that $S^e$ satisfies the conditions in the first case of Theorem \ref{theorem:LC of int seq},
that is $L ( T_{A}^{\mathrm{NTU}} ) = N$.
Then the upper and lower bounds on the linear complexity $L( S^e )$ are obtained as follows:
\begin{enumerate}
\item $L( S^e ) \leq 2 N - 1$.
The equality holds if and only if $G ( N, e ) = 1$.
\item $L( S^e ) \geq 2 N - N / 2^{\nu_2 ( N )}$.
The equality holds if and only if $- 2 e + 1 \equiv 0 \mathrm{~mod~} N / 2^{\nu_2 ( N )}$.
\end{enumerate}
\end{corollary}

\begin{example}
Let $p = 17$ and $m = 2$, so that $N = 288 = 2^5 \cdot 3^2$.
Assume that $( A / p ) = - 1$.
Then the linear complexity $L( S^e )$ of $S^e$ is given as
\begin{equation*}
L ( S^e ) = \left\{
\begin{array}{ll}
567 & \mbox{if~} e \equiv 5 \mathrm{~mod~} 9 \\
573 & \mbox{if~} e \equiv 2 \mathrm{~mod~} 9 \\
575 & \mbox{otherwise}.
\end{array}
\right.
\end{equation*}
\end{example}

\begin{example}
Let $p = 7$ and $m = 3$, so that $N = 342 = 2 \cdot 3^2 \cdot 19$.
Assume that $( A / p ) = - 1$.
Then the linear complexity $L( S^e )$ of $S^e$ is given as
\begin{equation*}
L ( S^e ) = \left\{
\begin{array}{ll}
344 & \mbox{if~} e \equiv 86 \mathrm{~mod~} 171 \\
452 & \mbox{if~} e \equiv 5 \mathrm{~mod~} 9, e \not\equiv 86 \mathrm{~mod~} 171 \\
458 & \mbox{otherwise}.
\end{array}
\right.
\end{equation*}
Note that $H_0 ( \nu, e )$ and $H_1 ( N, e )$ are given as in Tabel \ref{table:LC} if $e \equiv 5 \mathrm{~mod~} 9$.

\begin{table}[tb]
\caption{The values of $H_0 ( \nu, e )$ and $H_1 ( N, e )$ in the case of $e \equiv 5 \mathrm{~mod~} 9$.}
\label{table:LC}
\begin{center}
\begin{tabular}{ccc}
\hline\noalign{\smallskip}
$e$ & $H_0 ( \nu, e )$ & $H_1 ( N, e )$ \\
\noalign{\smallskip}\hline\noalign{\smallskip}
$e \equiv 86 \mathrm{~mod~} 171$ & $57$ & $171$ \\
$e \equiv 5 \mathrm{~mod~} 9, e \not\equiv 86 \mathrm{~mod~} 171$ & $3$ & $9$ \\
\noalign{\smallskip}\hline
\end{tabular}
\end{center}
\end{table}
\end{example}

\section{Conclusions}
\label{section:Conclusions}

In this paper, we first investigate linear complexity of generalized NTU sequences,
which are geometric sequences defined by cyclotomic classes.
In the case that the Chan--Games formula does not hold,
we show the new formula by considering the sequence of complement numbers, Hasse derivative and cyclotomic classes.
Under some conditions, we can ensure that the generalized NTU sequences have a large linear complexity
from the results on linear complexity of Sidel'nikov sequences.  

Next, we propose interleaved sequences of a binary generalized NTU sequence and its complement.
Furthermore, we show the periodic autocorrelation and linear complexity of the proposed sequences.
The proposed sequences have the balance property and
double the period of the generalized NTU sequences by the interleaved structure.
The autocorrelation distributions of the proposed sequences have a few peaks and troughs.
The question is whether a security is affected by these distributions.
We obtain the linear complexity of the proposed sequences
in the case that the generalized NTU sequences have a large linear complexity.
In particular, if the linear complexity of the generalized NTU sequence attains the maximal value,
the formula induces the upper and lower bounds on the linear complexity,
and the conditions to attain the upper and lower bounds, respectively.
Therefore, although the linear complexity do not attain the maximal value,
one can choose a shift width with which a sequence has a large linear complexity, especially such as the period minus one.
Thus the proposed sequences have good cryptographic properties for linear complexity.

In this paper, we only consider the frequency for a single symbol,
however the proposed sequences do not have the balance property for block of length larger than one.
Applying a linear transformation, one can expect to obtain well balanced sequence for block of a certain length.
However, the correlation properties are not kept.
As a future work, we should give a transformation,
by which transformed sequences have the balance property and keep some good properties of the propesed sequences.

\section*{Acknowledgments}

This research was supported by JSPS KAKENHI Grant-in-Aid for Scientific Research (A) Number 16H01723.

\bibliographystyle{abbrv}% bib style
\bibliography{GenNTU}% your bib database
%\begin{thebibliography}{99}% more than 9 --> 99 / less than 10 --> 9
%\bibitem{}
%\end{thebibliography}

\end{document}